\documentclass[11pt,leqno]{article}   

\usepackage[T1]{fontenc}
\usepackage[utf8]{inputenc}
\usepackage{authblk}
\usepackage[all]{xy}

\usepackage{amsmath,amstext,amsthm,amsfonts,amssymb,color,latexsym}
\numberwithin{equation}{section}
\def\h{\mathfrak H}
\def\C{\mathbb{ C}}

\def\R{\mathbb{ R}}

\usepackage{epsfig}
\newtheorem{thm}{Theorem}[section]

\newcommand{\beano}{\begin{eqnarray*}}
\newcommand{\enano}{\end{eqnarray*}}
\newcommand{\ena}{\end{eqnarray}}

\newtheorem{cor}[thm]{Corollary}

\newcommand{\be}{\begin{equation}}
\newcommand{\ee}{\end{equation}}
\newcommand{\en}{\end{equation}}

\newcommand{\ba}{\begin{array}}
\newcommand{\ea}{\end{array}}

\renewcommand{\o}{\overline}

\renewcommand{\em}{\it}

\newcommand{\la}{\lambda}

\newcommand{\bg}{\begin{gathered}}
\newcommand{\eg}{\end{gathered}}

\newcommand{\bea}{\begin{eqnarray}}
\newcommand{\eea}{\end{eqnarray}}
\newcommand{\Sum}{\sum_{m,n=0}^\infty}

\newcommand{\red}[1]{\textcolor[rgb]{0.9, 0, 0}{#1}}

\title{Deformed Complex Hermite Polynomials}
\author[a]{S. Twareque Ali
\thanks{Work supported in part by  the Natural Sciences and Engineering Research Council of Canada (NSERC)}}

\author[b]{Mourad E. H. Ismail
\thanks{Research supported by  the DSFP at King Saud University in Riyadh.}}

\author[c]{Nurisya M. Shah}
\affil[a]{Department of Mathematics and Statistics, Concordia University, Montr\'eal, Qu\'ebec, Canada H3G 1M8

email: {\it twareque.ali@concordia.ca}}
\affil[b]{Department of Mathematics, University of Central Florida,
Orlando, Florida 32816 USA 
and King Saud University, Riyadh, Saudi Arabia 

email: {\it mourad.eh.ismail@gmail.com}}

\affil[c]{Department of Physics, Faculty of Science, Universiti Putra Malaysia, 43000 UPM Serdang, Selangor, Malaysia

email:{\it{risyams@gmail.com}}}

\begin{document}

\maketitle
\begin{abstract}
 We study a class of bivariate deformed  Hermite polynomials and some of their properties using classical analytic techniques and the Wigner map. We also prove the positivity of certain determinants formed by the deformed polynomials. Along the way we also work out some additional properties of the (undeformed) complex Hermite polynomials and their relationships to the standard Hermite polynomials (of a single real variable). 
\end{abstract}



\section{Introduction}
The complex Hermite polynomials $\{H_{m,n}(z_1, z_2)\}$  may be  defined  by
\begin{equation}
\label{eqH-Ito}
H_{m,n}(z_1, z_2)=\sum_{k=0}^{m\wedge n}(-1)^k k!{m\choose k}{n\choose k}z_1^{m-k} z_2^{n-k}.
\end{equation}
Their exponential  generating function is
\begin{equation}
\label{eqH-GF}
\sum_{m,\,n=0}^{\infty}\,H_{m,n}(z_1, z_2)\frac{u^m\,v^n}{m!\,n!}=e^{uz_1+v  z_2 -uv}.
\end{equation}
 They satisfy the orthogonality relation,  \cite{Int:Int} and \cite{Gha}
 \begin{equation}
\label{H-Ito-OR}
\frac{1}{\pi}\int_{\R^2} H_{m,n}(x+iy,x-iy)\o{H_{p,q}(x+iy,x-iy)}\,e^{-x^2-y^2}\,dx\,dy
= m!\,n!\,\delta_{m,p}\,\delta_{n,q}.
\end{equation}
Many of their properties including a multilinear generating function are in \cite{Ism1} while their combinatorics have been studied in \cite{Ism:Sim} and \cite{Ism:Zen}.  New proofs of the Kibble-Slepian formula for the Hermite and Complex Hermite polynomials are in \cite{Ism:Zha}. The complex Hermite polynomials
 were introduced by Ito in \cite{Ito} and many of their properties have been developed in
 \cite{Ali:Bag:Hon,Cot:Gaz:Gor,Gha,Gha2,Ism1}, and \cite{Wun}--\cite{Wun2}.  They are in a class of polynomials
 presented in the very recent book by Dunkl  and Xu \cite[Chapter 2]{Dun:Xu}.

In this paper we  study the deformed complex Hermite polynomials $\{H_{m,n}^{(g)}(z_1, z_2)\}$. They  are defined through the generating function
\bea
\bg
\Sum \frac{u^m \, v^n}{m!\, n!} H_{m,n}^{(g)}(z_1, z_2) \qquad \qquad \qquad \qquad \\
= \exp((g_{1,1}u+ g_{1,2}v) z_1 + (g_{2,1}u+ g_{2,2}v) z_2 - (g_{1,1}u+ g_{1,2}v)(g_{2,1}u+ g_{2,2}v) ),
\eg
\label{eqGFHg}
\eea
where $g$ is the matrix
\bea
g :=  \left(\begin{matrix}
g_{1,1} & g_{1,2} \\
g_{2,1}  &  g_{2,2}
\end{matrix} \right).
\eea
A version of these polynomials have been studied in \cite{Wun,Wun2,Wun3}. In the setting adopted in this paper they were introduced in \cite{balshali}, where some preliminary properties, relating to orthogonality and growth were worked out as well as their relationship to a model of noncommutative quantum mechanics. 

In Section 2 we derive some of the properties of the deformed complex Hermite polynomials including their orthogonality relation, Rodrigues formula, and a moment representation. In Section 3 we show that certain Hankel determinants formed by the polynomials $\{H_{m,n}^{(g)}(z_1,  z_2)\}$ are nonnegative. This is done along the same lines of \cite{Bar:Ism,Ism2}, which was motivated by the earlier works of Karlin \cite{Kar1}--\cite{Kar2}, and the mammoth paper \cite{Kar:Sze} by Karlin and Szeg\H{o}.

\section{Some properties of $\{H_{m,n}^{(g)}(z_1,  z_2)\}$}
\begin{thm}\label{thm1}
Let $S^{(g)}$ be the operator defined by
\bea
(S^{(g)}f)(z_1, z_2) = f(g_{1,1} z_1 + g_{2,1}z_2, g_{1,2} z_1 + g_{2,2}z_2).
\eea
We have
\bea
H_{m,n}^{(g)}(z_1,  z_2) &=&
e^{-\partial_{z_1} \partial_{ z_2}}  S^{(g)} e^{\partial_{z_1} \partial_{z_2}}
H_{m,n}(z_1,  z_2),  \label{eq1.1.1}\\
H_{m,n}^{(g)}(z_1,  z_2)  &=& e^{-\partial_{z_1} \partial_{ z_2}} (g_{1,1}z_1+ g_{2,1}z_2)^m (g_{1,2}z_1+ g_{2,2}z_2)^n,  \label{eq1.1.2}\\
H_{m,n}^{(g)}(z_1,  z_2) &=& \sum_{j=0}^m \sum_{k=0}^{n} \binom{m}{j}\binom{n}{k}
g_{1,1}^j g_{2,1}^{m-j} g_{1,2}^k g_{2,2}^{n-k}  H_{j+k,n+m-j-k}(z_1,  z_2).
 \label{eq1.1.3}
\eea
\end{thm}
\begin{proof}
Multiply the right-hand side by $u^mv^n/(m!n!)$ and sum over $m,n \ge 0$ and use the generating function \eqref{eqH-GF}.  The result is
\bea
\bg
e^{-uv} e^{-\partial_{z_1} \partial_{ z_2}}  S^{(g)} e^{\partial_{z_1} \partial_{z_2}} e^{uz_1+vz_2}
=e^{-uv} e^{-\partial_{z_1} \partial_{ z_2}}  S^{(g)} e^{uv} e^{uz_1+vz_2} \\
= e^{-\partial_{z_1} \partial_{ z_2}}  \exp\left(u(g_{1,1}z_1+g_{2,1}z_2) +v(g_{1,2}z_1+ g_{2,2}z_2)\right)\\
= e^{-\partial_{z_1} \partial_{ z_2}} \exp\left(z_1(g_{1,1}u+g_{1,2}v) +z_2(g_{2,1}u+ g_{2,2}v)\right)\\
=  \exp\left(-(g_{1,1}u+g_{1,2}v)(g_{2,1}u+ g_{2,2}v)\right)
 \exp\left(z_1(g_{1,1}u+g_{1,2}v) +z_2(g_{2,1}u+ g_{2,2}v)\right)
\eg
\notag
\eea
and \eqref{eq1.1.1} follows. Similarly \eqref{eq1.1.2} follows from the generating function
\eqref{eqH-GF}.  Finally \eqref{eq1.1.3}  follows from \eqref{eq1.1.2}  and the binomial theorem.
\end{proof}

Note that the relation \eqref{eq1.1.2} is the Rodrigues formula for $H_{m,n}^{(g)}(z_1,  z_2)$. Also, as shown in \cite{balshali}, it is possible to rewrite (\ref{eq1.1.3}) in a somewhat different form for fixed $L= m+n$:
\begin{equation}
  H^{(g)}_{k, L-k}(z , \overline{z}) = \sum_{r = 0}^L M(g, L)_{rk} H_{r, L-r} (z, \overline{z}),
\label{al-form}
\end{equation}
where
\begin{equation}
M(g,L)_{rk} = \sum_{q=\max\{0,r+k-L\}}^{\min\{r,k\}}\binom{k}{q}
\binom{L-k}{r-q}g_{11}^{q}g_{21}^{k-q}g_{12}^{r-q}g_{22}^{L-k+q-r},  \quad 0\leq r,k \leq L\ .
\label{irrep-mat-elem}
\end{equation}

From \eqref{eq1.1.2} it  is clear that
\bea
H_{m,n}^{(g)}(z, \overline{z}) = H_{n,m}^{(h)}(z, \bar z), \textup{where}\;
h =  g = \left(\begin{matrix}
0 & 1 \\
1  & 0
\end{matrix} \right) .
\eea
We now establish the orthogonality relation (see also \cite{balshali,Wun3}).
\begin{thm}
The orthogonality relation
\bea
\notag
\int_{\R^2} H_{m,n}^{(g)}(z, \bar z) \overline{H_{p,q}^{(h)}(z, \bar z)} e^{-x^2-y^2} dxdy = 0
\eea
if $(m,n) \ne (p,q)$ holds if and only if
\bea
\label{eqcond}
h^* g =  \left(\begin{matrix}
\la_1 & 0 \\
0  & \la_2
\end{matrix} \right).
\eea
When \eqref{eqcond} holds then
\bea
\label{eqorth}
\int_{\R^2}  H_{m,n}^{(g)}(z, \bar z) \overline{H_{p,q}^{(h)}(z, \bar z)} \,  e^{-x^2-y^2}\, dx dy = m!\, n!\,  \la_1^m \, \la_2^n\,  \delta_{m,p}\, \delta_{n,q}.
\eea
\end{thm}
\begin{proof}
To save space we denote rows 1 and 2 of
\bea
 \left(\begin{matrix}
g_{1,1} & g_{1,2} \\
g_{2.1} &g_{2,2}
\end{matrix} \right)
 \left(\begin{matrix}
u \\
v
\end{matrix} \right)
\notag
\eea
by $g_1(u,v)$ and $g_2(u,v)$, respectively.  It is straight forward to use the generating function
\eqref{eqGFHg} and see  that for  real $u_j, v_j, j=1,2$ we have
\bea
\notag
\bg
\sum_{m,n=0}^\infty \frac{u_1^m  v_1^n u_2^p v_2^q}{m!\, n!\,p!\, q!} \int_{\R^2} H_{m,n}^{(g)}(z, \bar z) \overline{H_{p,q}^{(h)}(z, \bar z)} e^{-x^2-y^2} dxdy  \\
= \int_{\R^2}   \exp(g_1(u_1,v_1) z  + g_2(u_1,v_1) \bar z - g_1(u_1,v_1)g_2(u_1,v_1)) \\\overline{\exp(h_1(u_2,v_2) z  + h_2(u_2,v_2) \bar z - h_1(u_2,v_2) h_2(u_2,v_2))}
 e^{-x^2-y^2} dxdy \\
 = \exp(-g_1(u_1,v_1)g_2(u_1,v_1) - \o{ {h_1}}(u_2,v_2) \o {h_2}(u_2,v_2)) \\
 \times
  \int_{\R^2}   \exp((g_1(u_1,v_1)+  \o {h_2}(u_2,v_2))z +(g_2(u_1,v_1)  + \o{h_1}(u_2,v_2))      \bar z)  e^{-x^2-y^2} dxdy.
\eg
\eea
By evaluating the integral we see that the integral in the last line  is
\bea
\bg
 \exp\left((g_1(u_1,v_1)+  \o {h_2}(u_2,v_2))(g_2(u_1,v_1)  + \o{h_1}(u_2,v_2))\right)
\eg
\notag
\eea
We have orthogonality if and only if
\bea
\notag
\sum_{m,n=0}^\infty \frac{u_1^m  v_1^n u_2^p v_2^q}{m!\, n!\,p!\, q!} \int_{\R^2} H_{m,n}^{(g)}(z, \bar z) \overline{H_{p,q}^{(h)}(z, \bar z)} e^{-x^2-y^2} dxdy = f(u_1u_2, v_1v_2),
\eea
for some function $f$ of two variables.  This is equivalent to the condition \eqref{eqcond}.
\end{proof}

It is clear from \eqref{eqorth} that we can rescale $H_{m,n}^{(g)}(z, \bar z)$ and
$H_{m,n}^{(h)}(z, \bar z)$ to make $\la_1 = \la_2=1,$ which we now assume. Therefore we assume
that
\bea
h = (g^*)^{-1}.
\eea
Thus we have the orthogonality relation
\bea
\label{eqorthrel}
\int_{\R^2} H_{m,n}^{(g)}(z, \bar z) \overline{H_{p,q}^{(h)}(z, \bar z)} e^{-x^2-y^2} dxdy = m!n! \delta_{m,p}\delta_{n,q},
\eea
 where $h = (g^*)^{-1}.$

\begin{thm}
Let $z = x+iy$. The polynomials $\{H_{m,n}^{(g)}(z_1,  z_2)\}$ have the integral representation
\bea
H_{m,n}^{(g)}(iz,  i \bar z) = \frac{i^{m+n}}{\pi}\int_{\R^2}(g_{1,1} \zeta  + g_{2,1}\o\zeta)^m
 (g_{1,2} \zeta + g_{2,2}\o\zeta)^n  e^{-(r-x)^2-(s-y)^2}\, dr ds,
\label{eqmomentrep}
\eea
where $\zeta := r+ i s$.
\end{thm}
\begin{proof}
First replace $\zeta$ by $\zeta + z$ in the right-hand side of \eqref{eqmomentrep} then
multiply the right-hand side of \eqref{eqmomentrep} by $u^m v^n/(m!n!)$ and add the terms for
$m,n \ge 0$. This sum equals
\bea
\notag
\bg
\frac{1}{\pi}  \exp\left( i u(g_{1,1} z  + g_{2,1}\bar z) +iv (g_{1,2} z + g_{2,2}\bar z)\right) \\
\times  \int_{\R^2} \exp\left(-r^2 - s^2 + i u(g_{1,1} \zeta  + g_{2,1}\o\zeta) +iv (g_{1,2} \zeta + g_{2,2}\o\zeta)\right) \, drds.
\eg
\eea
The integral in the above expression is given by
\bea
\notag
\bg
\int_{\R^2} \exp\left(-r^2 - s^2 + i ur(g_{1,1}    + g_{2,1}) + ivr(g_{1,2}+ g_{2,2}) + us( g_{2,1}) -g_{1,1})
 +vs (g_{2,2}-g_{1,2} )\right) \, drds \\
 =  \exp\left(   (u( g_{2,1} -g_{1,1})
 +v (g_{2,2}-g_{1,2}))^2/4 -  (u(g_{1,1}    + g_{2,1}) + v(g_{1,2}+ g_{2,2}))^2/4\right) \\
 = \exp\left( (ug_{1,1} + v g_{1,2})(ug_{2,1} + v g_{2,2})\right).
 \eg
\eea
Therefore the exponential generating function of the right-hand side of \eqref{eqmomentrep} simplifies
to the exponential generating function of the left-hand side as in \eqref{eqGFHg} and the proof is complete.
\end{proof}

For many applications of moment techniques to special functions we refer the reader to
\cite{Ism}.

\section{A positivity result}

In this section we establish the positivity of a Hankel determinant formed by the $H_{m,n}^g$ polynomials.

\begin{thm} \label{thm2}
Assume that $g_{1,2} = \overline{g_{2,1}}$  and  $g_{2,2} = \overline{g_{1,1}}$.
Then the  determinant formed by $(-i)^{m+n}(-1)^sH^{(g)}_{m+s,n+s}(iz,i\bar z): 0\le m,n \le N$ is positive for all $N \ge 0$.
\end{thm}
\begin{proof}
Let $\Delta_N$ be the determinant whose elements  are
$$(-i)^{m+n+2s} \pi H_{m+s,n+s}(iz,i\bar z): 0 \le m,n < N.$$
It is convenient to set
\bea
U_k = g_{1,1}\zeta_k+ g_{2,1}\o{\zeta_k}, \quad V_k = \o{U_k} =
g_{1,2}\zeta_k + g_{2,2}\o{\zeta_k}.
\eea
The integral representation \eqref{eqmomentrep} implies that  $\Delta_N$ has the integral representation, with $\zeta_k = r_k + i s_k$,
\bea
\notag
\bg
 \int_{\R^{2N}}  \begin{vmatrix}
1  &V_1& \dotsm & (V_1)^{N-1} \\
U_2 & U_2V_2   & \dotsm & U_2 (V_2)^{N-1}  \\
\vdots & \vdots & \dotsm & \vdots         \\
 (U_N)^{N-1}  &  (U_N)^{N-1}V_N & \dotsm & (U_N)^{N-1}(V_N)^{N-1}
\end{vmatrix} \\
\times  \prod_{k=1}^N (U_k)^s
(V_k)^s\;  \prod_{j=1}^N e^{-(r_j-x)^2-(s_j-y)^2}\, dr_j\, ds_j\\
=  \int_{\R^{2N}}  \prod_{k=1}^N |U_k|^{2s}
 \prod_{j=1}^N (U_j)^{j-1}
 \begin{vmatrix}
1  & V_1  & \dotsm & (V_1)^{N-1} \\
 1 &  V_2  & \dotsm &   (V_2)^{N-1}  \\
\vdots & \vdots & \dotsm & \vdots         \\
 1  &   V_N & \dotsm &  (V_N)^{N-1}
\end{vmatrix}\\
\times  \prod_{j=1}^N e^{-(r_j-x)^2-(s_j-y)^2}\,  dr_j\, ds_j.
\eg
\eea
Now apply a permutation on  the tuples $(r_j,s_j)$. If reorder the Vandermonde determinant with indices in increasing order we should multiply  the answer by the sign of the permutation,
say sign$(\sigma)$.  We then sum over $\sigma$ in the symmetric group $S_N$ and divide by the size of $S_N$, that is we divide by $N!$\red{.}  Therefore
\bea
\Delta_N = \frac{1}{N!}  \int_{\R^{2N}}\; \prod_{k=1}^N |U_k|^{2s}\;
\left[ \prod_{1 \le j < k \le N} \left|U_j- U_k\right|^2\right]
\prod_{j=1}^N dr_j\, ds_j.
\label{eqDeltaN}
\eea
\end{proof}
\begin{cor}
Set $z = x+iy$ and let $\Delta_N$ be the determinant whose elements are
$$(-i)^{m+n} \pi H_{m,n}(iz,i\bar z): 0 \le m,n < N.$$
Then $\Delta_N$ is given by \eqref{eqDeltaN}.
\end{cor}

\section{The Wigner map and the polynomials $\{H_{m,n}\}$}\label{wigner-map}
To proceed further, we write the (undeformed) polynomials $H_{m,n}$ as basis elements of the
Hilbert space $L^2(\mathbb C, d\nu (z, \overline{z}))$, where $d\nu(z, \overline{z})= e^{-\vert z\vert^2}\dfrac {dx\; dy}{\pi}$. It is well known (see, for example \cite{Gha}) that the vectors $h_{m,n} = \dfrac {H_{m,n}}{\sqrt{m!n!}}$ form an orthonormal basis of $L^2(\mathbb C, d\nu (z, \overline{z}))$ and they may be obtained by the action of two differential operators on the
{\em ground state} vector $h_{0,0} \in
L^2(\mathbb C, d\nu (z, \overline{z}))$, which is the constant function, taking the
value $1$ everywhere:
\be
  h_{m,n} (z, \overline{z}) = \frac {(z - \partial_{\overline{z}})^m\;
  (\overline{z} - \partial_z)^n}{\sqrt{m!n!}}\; h_{0,0}\; , \qquad m,n =0,1,2,
  \ldots , \infty\; .
\label{compherm-basis}
\en
Furthermore, the operators
\be
 a_1 = \partial_z , \;\; a_1^\dag = z - \partial_{\overline{z}}, \qquad
 a_2 = \partial_{\overline{z}} , \;\; a_2^\dag = \overline{z} - \partial_z,
\label{comp-ccr}
\en
satisfying the commutation relations,
\be
  [a_i, a^\dag_j ] = \delta_{ij}, \quad i,j = 1,2, \qquad [a_1 , a_2] =0,
\label{comp-ccr2}
\en
form an irreducible set on $L^2(\mathbb C, d\nu (z, \overline{z}))$. In terms of these operators,
\be
  h_{m,n} (z, \overline{z}) = \frac {(a_1^\dag)^m\;
  (a_2^\dag)^n}{\sqrt{m!n!}}\; h_{0,0}\; , \qquad m,n =0,1,2,
  \ldots , \infty\; .
\label{compherm-basis2}
\en

We next define the {\em Wigner map} (see, for example, \cite{aag-book}) in an abstract setting. Let $\h$ be an infinite dimensional, abstract, separable Hiblert space over the complexes and let $\{\phi_n\}_{n=0}^\infty$ be an orthonormal basis of it. Define the two (annihilation and creation) operators $a, a^\dag$,
\be
  a\phi_n = \sqrt{n}\phi_{n-1}, \quad a\phi_0 = 0, \qquad a^\dag\phi_n = \sqrt{n+1}\phi_{n+1}\; .
\label{crannop}
\ee
Then,
\be
  \phi_n = \frac {(a^\dag)^n}{\sqrt{n!}}\phi_0, \qquad n =0,1,2, \ldots , \infty\; .
\label{onbasis}
\en

For $z \in \mathbb C$ define the unitary {\em displacement} operator
\be
   D(z, \overline{z}) = e^{za^\dag - \overline{z}a} =  e^{za^\dag}  e^{-\overline{z}a} e^{-\frac {\vert z\vert^2}2} = e^{-\overline{z}a}e^{za^\dag}e^{\frac {\vert z\vert^2}2} \;.
\label{disp-op}
\ee
Let $\mathcal B_2 (\h)$ denote the set of all Hilbert-Schmidt operators on $\h$. It is well known that this is a Hilbert space under the scalar product
$$ \langle X \mid Y\rangle_{\mathcal B_2} = \text{Tr}[X^* Y], \qquad X, Y \in
\mathcal B_2 (\h)\; ,$$
and the vectors $\{\vert\phi_n\rangle\langle\phi_m\vert, \;\; m,n = 0,1,2, \ldots , \infty\}$, form an orthonormal basis of $\mathcal B_2 (\h)$.
The Wigner map is then defined as the linear transformation, $\mathcal W : \mathcal B_2 (\h) \longrightarrow L^2(\mathbb C, d\nu (z, \overline{z}))$:
\be
   (\mathcal W X)(z, \overline{z}) = e^{\frac{\vert z\vert^2}2}\;\text{Tr}\; [D(z, \overline{z})^* X]\;, \qquad X \in \mathcal B_2 (\h)\; .
\label{wigmap}
\ee
This map is well known to be a linear isometry, i.e.
$$ \Vert \mathcal W X\Vert^2_{L^2} = \int_{\mathbb C}\vert (\mathcal W X)(z, \overline{z})\vert^2\; d\nu (z, \overline{z}) = \text{Tr}[X^* X] = \Vert X\Vert^2_{\mathcal B_2}\; , $$
(being first defined on finite rank elements $X$ and then extended by continuity).

For any operator $B$ on $\h$, let us define the two commuting operators, $B_\ell, \; B_r$ on
$\mathcal B_2 (\h)$:
\be
B_\ell (X) = BX, \qquad B_r (X) = XB^*, \qquad X \in \mathcal B_2 (\h)\; ;
\label{left-right-ops}
\en
we shall be interested in their Wigner transformed versions, $\mathcal W \{B_\ell, \; B_r\}\mathcal W^*$ on the Hilbert space $L^2(\mathbb C,d\nu (z, \overline{z}))$. Using (\ref{disp-op}) to write (\ref{wigmap}) in the form
$$
(\mathcal W X)(z, \overline{z}) = \text{Tr}[e^{-za^\dag}Xe^{\overline{z} a}]\;
    e^{z\overline{z}} = \text{Tr}[e^{\overline{z}a}Xe^{-{z} a^\dag}]
  $$
and differentiating with respect to $z$ and $\overline{z}$, we get for the Wigner
transformed versions of the operators $a_\ell, \; a_\ell^\dag,\; a_r, \; a_r^\dag$,
corresponding to the $a, \; a^\dag$ in (\ref{crannop}),
\bea
\mathcal W a_\ell \mathcal W^* = \partial_{\overline{z}} = a_2 \; ,&\qquad&
 \mathcal W a^\dag_\ell \mathcal W^* = \overline{z} - \partial_z = a_2^\dag\; ,\nonumber\\
 \mathcal W a_r \mathcal W^* = -\partial_z = -a_1 \; ,&\qquad&
 \mathcal W a^\dag_r \mathcal W^* = - z + \partial_{\overline{z}} = -a_1^\dag\; .
 \label{wig-tr-cran}
 \eea
 From this, and since $\mathcal W (\vert\phi_0\rangle\langle \phi_0\vert)= h_{0,0}$, we easily get
 \be
 \mathcal W (\vert\phi_n\rangle\langle \phi_m\vert) = \frac {(-a_1^\dag)^m\; (a_2^\dag)^n}{\sqrt{m!\;n!}}h_{0,0} = (-1)^m h_{m,n}\; .
\label{wig-trans-ccomp-her}
 \en

   The Hilbert space $\h$ was taken to be an abstract space. As a concrete realization let us take
it to be the space $L^2(\mathbb R, \dfrac{e^{-x^2}\;dx}{\sqrt{\pi}})$, on which $a = \dfrac 1{\sqrt{2}}\partial_x$ and
$a^\dag = \dfrac 1{\sqrt{2}}(2x - \partial_x)$. Also, $\phi_0$ is the constant vector which is equal to one everywhere and
\be
  \sqrt{2^n\;n!}\;\phi_n (x) = (2x - \partial_x)^n \phi_0 (x) = H_n (x),
\label{real-herm}
\en
where the $H_n(x)$ are the real Hermite polynomials, \cite{Ism},
$$
   H_n(x) = (-1)^n e^{x^2}\frac {\partial^n}{\partial x^n} e^{-x^2}\; . $$
Furthermore, $\mathcal B_2 (\h )$ can now be identified with $L^2(\mathbb R^2,
\dfrac{e^{-x^2 -y^2}\;dx\;dy}{\pi})$, so that
$$(\vert\phi_n\rangle\langle\phi_m\vert) (x,y) = \frac 1{\sqrt{2^{n+m}\;n!\;m!}}H_m (x)H_n (y)\; .$$
Thus, from (\ref{wigmap}) and (\ref{wig-trans-ccomp-her})
\beano
  H_{m,n} (z , \overline{z} )  & = & \frac {(-1)^n}{\sqrt{2^{m+n}}}\mathcal W (\vert H_n\rangle
             \langle H_m\vert )(z , \overline{z} )  \\
             & = & \frac {(-1)^n e^{\frac {\vert z \vert^2}2}}{\sqrt{2^{m+n}}}\; \text{Tr} [D (z, \overline{z})^* \vert H_n\rangle\langle H_m\vert]
\enano
To compute the trace in the above formula explicitly, we not that  (see, e.g., \cite{aag-book})  for
any $f \in L^2(\mathbb R, \dfrac{e^{-x^2}\;dx}{\sqrt{\pi}})$,
$$
   [D (z, \overline{z})^* f](u) = e^{-i(x + \sqrt{2}u)y}f(u + \sqrt{2}x), \qquad z = x+iy. $$
We thus have the result,
\begin{thm} \label{comp-real-hermite-thm}
The complex Hermite polynomials are  Wigner transforms of bilinear products of the real Hermite polynomials:
\bea
H_{m,n} (z , \overline{z} )& = & \frac {(-1)^n}{\sqrt{2^{m+n}}}\mathcal W (\vert H_n\rangle
             \langle H_m\vert )(z , \overline{z} ) \nonumber \\
   & = & \frac {(-1)^n}{\sqrt{2^{m+n}}}\int_{\mathbb R} e^{-(u + \frac {\overline{z}}{\sqrt{2}})^2} H_m (u) H_n (u + \frac 1{\sqrt{2}} ( z + \overline{z})) \; \frac {du}{\sqrt{\pi}}\; .
\label{comp-real-herm}
\eea
In fact we have the more general result
\bea
H_{m,n} (z , \overline{z}) =  \frac{(-1)^n}{2^{(m+n)/2}}\int_\R e^{-(u+a)^2} H_m(u+b) H_n(u+c) \frac{du}{\sqrt{\pi}},
\label{eqHmnInt}
\eea
where
\bea
b = a + z/\sqrt{2},\qquad  a-c = \bar z/\sqrt{2}.
\eea
\end{thm}
\begin{proof} We give another proof using generating functions.   Consider the integral
\bea
I(m,n) = \frac{1}{2^{(m+n)/2}}\int_\R e^{-(u+a)^2} H_m(u+b) H_n(u+c) \frac{du}{\sqrt{\pi}}.
\eea
Clearly
\bea
\bg
\sum_{m,n=0}^\infty I(m,n) \frac{s^mt^n}{m!\; n!}
= \frac{1}{\sqrt{\pi}}\int_\R \exp\left(-(u+a)^2 + \sqrt{2}s(u+b) + \sqrt{2}t(u+c) - (s^2+t^2)/2\right) du \\
=  \frac{1}{\sqrt{\pi}}\int_\R \exp\left(-(u+a -(s+t)/\sqrt{2})^2 + \sqrt{2}sb + \sqrt{2}tc  +st
- \sqrt{2} a(s+t)  \right) \\
= \exp\left( \sqrt{2}s(b-a) + \sqrt{2}t(c-a)  +st \right).
\eg
\notag
\eea
We now make the choices indicated in the theorem and prove \eqref{eqHmnInt}.
\end{proof}

From (\ref{al-form}) it is clear that (\ref{eqHmnInt}) can be extended to deformed polynomials in the manner
\bea
H^{(g)}_{k,L-k} (z , \overline{z}) =  \frac{(-1)^n}{2^{(m+n)/2}}\sum_{r=0}^L M(g, L)_{rk}\int_\R e^{-(u+a)^2} H_r (u+b) H_{L-r}(u+c) \frac{du}{\sqrt{\pi}},
\label{eqHmnInt-def}
\eea

The integral representation \eqref{eqHmnInt} suggests a possible extension to two variables.
Indeed, we have the following result.
\begin{thm}
For $a \in \C$, we have the integral representation
\bea
\label{eqeDHIR}
\bg
  \int_{\R^2} H_{m,n}(z+a, \bar z + \bar a)H_{p,q}(z-a, \bar z- \bar a) e^{-z \bar z} dx dy
 \\
 = (-1)^{p+q} \pi  H_{m,q}(a,\bar a)
 H_{n,p}(a,\bar a).
\eg
\eea
\end{thm}
\begin{proof}
Let $f(m,n,p,q)$ denote the right-hand side of \eqref{eqeDHIR}. The generating function
\eqref{eqH-GF} implies
\bea
\bg
\notag
\sum_{m,n,p,q=0}^\infty f(m,n,p,q) \; \frac{u^m\,v^n\, \xi^p\, \eta^q}{m!\, n!\,p!\, q!} \\
= e^{au+ \bar a v +  a \xi+ \bar a \eta -uv- \xi\eta}  \int_{\R^2}
\exp\left((x(u+v+ \xi+\eta) +iy(u-v+\xi-\eta\right)
e^{-x^2-y^2}\, dxdy \\
=\pi \exp\left((a (u -\xi) + \bar a (v-\eta)  -uv-\xi\eta   + (u+\xi)(v+\eta)\right) \\
= \pi \exp\left((a (u-  \xi) + \bar a (v - \eta) + v \xi    + u \eta \right) \\
= \pi \exp\left((a u + \bar a ( - \eta) +  a(-\xi) + \bar a  v + v \xi    + u \eta  \right).
\eg
\eea
In view of the generating function \eqref{eqH-GF} the above expression is
\bea
\notag
\pi \sum_{m,q=0}^\infty H_{m,q}(a,\bar a) \frac{u^m (-\eta)^q}{m!q!}
\sum_{n,p=0}^\infty H_{n,q}(a,\bar a) \frac{v^n (-\xi)^p}{n!p!}
\eea
This establishes the theorem.
\end{proof}

It is clear from (\ref{eqH-Ito}) that $H_{m,n}(0,0) = (-1)^n n!\; \delta_{m,n}$. Thus, the case of $a=0$ in (\ref{eqeDHIR}) gives the orthogonality relation. Furthermore, using (\ref{al-form}) we could extend (\ref{eqeDHIR}) to the case of deformed Hermite polynomials in a straightforward manner.

\section{A second representation}\label{sec-sec-rep}
A second representation of the $H_{m,n}$ in terms of the real Hermite polynomials can be obtained by
 noting that $L^2(\mathbb C , d\nu (z, \overline{z})\simeq L^2(\mathbb R, \dfrac{e^{-x^2}\;dx}{\sqrt{\pi}})\otimes L^2(\mathbb R, \dfrac{e^{-y^2}\;dy}{\sqrt{\pi}})$. Then writing
 $$ z - \partial_{\overline{z}} = \frac 12 (2x - \partial_x ) + \frac i2 (2y - \partial_y ), \qquad
    \overline{z} - \partial_z = \frac 12 (2x - \partial_x ) - \frac i2 (2y - \partial_y ), $$
in (\ref{compherm-basis})  we get
$$ H_{m,n} (z, \overline{z}) = \frac 1{2^{m+n}}[(2x - \partial_x )  + i(2y - \partial_y )]^m\;
       [(2x - \partial_x )  - i(2y - \partial_y )]^n\; h_{0,0}. $$
Expanding the binomials and comparing with (\ref{real-herm}),
\be
   H_{m,n} (z, \overline{z}) = \frac 1{2^{m+n}} \sum_{k=0}^m\sum_{\ell =0}^n (-1)^{n-\ell} (i)^{m+n-k-\ell} \begin{pmatrix} m\\ k \end{pmatrix}\;\begin{pmatrix} n\\ \ell\end{pmatrix}H_{k+ \ell}(x) H_{m+n-k-\ell} (y)\; .
\label{comp-real-herm3}
\en
This result has also been obtained in \cite{Ism:Zha}, using a different technique.

  It is useful to rewrite the above expression in a somewhat different manner. Fixing $m+n = L$ and writing $k+\ell = q$ the above equation can be written as (see (\ref{irrep-mat-elem}))
\bea
   H_{m,L-m} (z, \overline{z})& = & \frac 1{2^L} \sum_{r=0}^L \left[\sum_{k = \max\{0,r+m-L\}}^{\min\{r,m\}} (-1)^{L-m-r+k} (i)^{L-r} \begin{pmatrix} m\\ k \end{pmatrix}\;\begin{pmatrix} L-m\\ r-k\end{pmatrix}\right]\nonumber\\
    & \quad & \times H_r(x) H_{L-r} (y)\; .
\label{comp-real-herm5}
\ena
Defining the $(L+1)\times (L+1)$ matrix $M(L)$, with elements
$$
   M(L)_{r m} = \frac 1{2^L}\sum_{k = \max\{0,r+m-L\}}^{\min\{r,m\}} (-1)^{L-m-r+k} (i)^{L-r} \begin{pmatrix} m\\ k \end{pmatrix}\;\begin{pmatrix} L-m\\ r-k\end{pmatrix}, $$
we can also write (\ref{comp-real-herm5}) as
\be
   H_{m,L-m} (z, \overline{z}) = \sum_{r=0}^L M(L)_{r m}H_r(x) H_{L-r} (y)\; .
\label{comp-real-herm7}
\en
Once again, using (\ref{al-form}) we could extend the above relation  to the case of deformed Hermite polynomials.

We can interpret Eq.(\ref{comp-real-herm7}) in the following manner. Both sets, $H_m (x)H_n(y), \; m,n = 0,1,2, \ldots \infty$, and $H_{m,n} (z. \overline{z}), \; m,n = 0,1,2, \ldots \infty$,
form bases in the Hilbert space
$$
L^2(\mathbb R, \dfrac{e^{-x^2 - y^2}\;dxdy}{\pi}) \simeq
L^2(\mathbb R, \dfrac{e^{-x^2}\;dx}{\sqrt{\pi}})\otimes L^2(\mathbb R, \dfrac{e^{-y^2}\;dy}{\sqrt{\pi}}) $$
For fixed $L$, the set of vectors $H_r\otimes H_{L-r}, \; r =0,1,2, \ldots L$, span an $(L+1)$-dimensional subspace, $\mathfrak H(L)$, of this Hilbert space. Two subspaces, $\mathfrak H(L)$ and $\mathfrak H(L^\prime)$ are orthogonal if $L \neq L'$. The matrix $M(L)$ simply effects a basis change on this subspace, to the new basis $H_{r, L-r}, \; r =0,1,2, \ldots L$. The integral transform
(\ref{comp-real-herm}) effects a basis change on the entire Hilbert space. Denoting this transform by the operator $M$,  so that $M(H_m\otimes H_n) = H_{m,n}, \; m,n = 0,1,2, \ldots , \infty$, we see that it decomposes into the orthogonal direct sum $M = \oplus_{L=0}^\infty M(L)$. General maps of the type $M(L)$ have been
studied in \cite{balshali} and shown to be useful in  the construction of {\em pseudo fermions} in \cite{albaga}.


\end{document}